\documentclass[a4paper]{article}
\usepackage{pdfsync}
\usepackage[utf8]{inputenc}
\usepackage{graphicx}
\usepackage{amssymb,amsfonts,amsmath,amsthm}
\usepackage{hyperref}
\usepackage{graphicx}
\usepackage{multirow}
\usepackage{verbatim}
\usepackage{cite}
\usepackage{color}

\theoremstyle{plain}
\newtheorem{theorem}{Theorem}[section]
\newtheorem{corollary}{Corollary}[section]

\newtheorem{proposition}{Proposition}[section]
\theoremstyle{definition}
\newtheorem{definition}{Definition}[section]

\newtheorem{remark}{Remark}[section]

\begin{document}

\title{\textbf{Quantifying non-periodicity of non-stationary time series through wavelets}}


\author{V. J. Bol\'os$^1$, R. Ben\'{\i}tez$^1$, R. Ferrer$^2$ \\ \\
{\small $^1$ Dpto. Matem\'aticas para la Econom\'{\i}a y la Empresa, Facultad de Econom\'{\i}a,} \\
{\small Universidad de Valencia. Avda. Tarongers s/n, 46022 Valencia, Spain.} \\
{\small e-mail\textup{: \texttt{vicente.bolos@uv.es} (V. J. Bol\'os), \texttt{rabesua@uv.es} (R. Ben\'{\i}tez)}} \\ \\
{\small $^2$ Dpto. Econom\'{\i}a Financiera y Actuarial, Facultad de Econom\'{\i}a,} \\
{\small Universidad de Valencia. Avda. Tarongers s/n, 46022 Valencia, Spain.} \\
{\small e-mail\textup{: \texttt{roman.ferrer@uv.es}}} \\}

\date{June 2019}

\maketitle

\begin{abstract}
In this paper, we introduce a new wavelet tool for studying the degree of non-periodicity of time series that is based on some recently defined tools, such as the \textit{windowed scalogram} and the \textit{scale index}. It is especially appropriate for non-stationary time series whose characteristics change over time and so, it can be applied to a wide variety of disciplines. In addition, we revise the concept of the scale index and pose a theoretical problem: it is known that if the scale index of a function is not zero then it is non-periodic, but if the scale index of a function is zero, then it is not proved that it has to be periodic. This problem is solved for the particular case of the Haar wavelet, thus reinforcing the interpretation and applicability of the scale index as a useful tool for measuring non-periodicity. Finally, we discuss the relationship between non-periodicity and unpredictability, comparing the new wavelet tool with the sample entropy.
\end{abstract}


\section{Introduction}
\label{intro}

A large part of the analysis of a time series or signal consists of the determination of its main features and whether those features remain constant over time, so that this series can be considered stationary. A very important type of stationary time series are the periodic signals. In fact, it is obvious that the more periodic a signal is, the more easy and reliable the forecasts of this signal will be. However, in practice, a time series will be rarely periodic because, for example, it may be affected by a random noise, it may be a quasi-periodic signal arising from a dynamical system, or even a chaotic signal.

The \textit{scale index} is a wavelet tool designed to measure the degree of non-periodicity of a time series through its wavelet scalogram, allowing to quantify how much chaotic a signal is. This tool was introduced in~\cite{Ben10} and it has been used in a wide variety of recent works applied in many different scientific disciplines, such as the study of speech signals \cite{Hes13}, pseudo random number generators \cite{Akh14,Ava15,Yan16a,Yan16b,Tun16,Mur16}, images encryption \cite{Yan15,Yan16a}, meteorology \cite{Qib13}, biomedicine \cite{Beh13}, robotics \cite{Fel15}, engineering \cite{Pic15}, mechanical fault identification \cite{Seu15}, etc.

The scale index is based on a result, proved in~\cite{Ben10}, which states that the continuous wavelet transform of a periodic function, associated with a compactly supported wavelet, vanishes when the scale is twice the period (for all times). From this result, the scale index vanishes for all periodic signals and depends continuously on the signal itself. Therefore, when a signal is close to a periodic one, its scale index will take a value close to zero.

Nevertheless, the scale index simply provides a global or average measure of non-periodici\-ty for the entire time series, and hence it does not consider the time location. This fact could be a serious handicap for the study of non-stationary signals whose behaviour varies substantially over time and, consequently, their non-periodicity may not be constant. For these cases, we introduce the \textit{windowed scale index}, that is based on the \textit{windowed scalogram} (see \cite{Bol17}) and the scale index. This new wavelet tool inherits the properties of the scale index for measuring the non-periodicity of a time series centered at a given point with a given time radius, preserving in this way the time location. Moreover, using the time radius as a parameter makes this tool much more versatile.

This paper is organized as follows. Section \ref{sec:back} introduces the wavelet mathematical background, revises the concept of the scale index and presents a new theoretical result (see Proposition \ref{prop2.1}) that reinforces the interpretation of the scale index (windowed or not) as a non-periodicity measure. In Section \ref{sec:wsi} we define the windowed scale and, in Section \ref{sec:ex}, we give examples and an application in economics. Moreover, we discuss the relationship between non-periodicity and unpredictability, comparing the windowed scale index with the \textit{sample entropy}, a well-known complexity measure. Finally, Section \ref{sec:conc} provides some concluding remarks.

\section{The scale index revisited}
\label{sec:back}

\subsection{Basic concepts of wavelets}
\label{prelim}

A \textit{wavelet} is a function $\psi \in L^2\left( \mathbb{R}\right) $ with zero average (i.e.  $\int _{\mathbb{R}} \psi =0$), normalized ($\| \psi \| =1$) and ``centered'' in the neighborhood of $t=0$ (see \cite{mal98}). Scaling $\psi$ by $s>0$ and translating it by $u\in\mathbb{R}$, we can create a family of \textit{time--frequency atoms} (also called \textit{daughter wavelets}), $\psi_{u,s}$, as follows
\begin{equation}
\label{eq:psius}
\psi _{u,s}(t):=\frac{1}{\sqrt{s}}\psi \left( \frac{t-u}{s}\right) .
\end{equation}

Given a time series $f\in L^2\left( \mathbb{R}\right) $, the \textit{continuous wavelet
transform} (CWT) of $f$ at time $u$ and scale $s$ with respect to the wavelet $\psi $ is
defined as
\begin{equation}
\label{eq:cwt}
Wf\left( u,s\right) :=\left<f,\psi _{u,s}\right> =\int _{-\infty}^{+\infty}f(t)\psi ^*
_{u,s}(t) \, \textrm{d}t,
\end{equation}
where $ ^*$ denotes the complex conjugate. The CWT allows us to obtain the frequency
components (or \textit{details}) of $f$ corresponding to scale $s$ and time location $u$,
thus providing a time-frequency decomposition of $f$.

The \textit{scalogram} of a time series $f$ at a given scale $s>0$ is given by
\begin{equation}
\label{scs}
\mathcal{S}(s):= \left( \int _{-\infty } ^{+\infty } | Wf\left( u,s\right) | ^2 \,
\textrm{d}u \right) ^{1/2} .
\end{equation}
The scalogram\footnote{We use the definition of scalogram given in \cite{Ben10} and \cite{Bol17}, that is not the one commonly acknowledged: usually it is the squared magnitude of the CWT.} of $f$ at $s$ is the $L^2$-norm of $Wf\left( u,s\right)$ (with respect to the time variable $u$) and captures the ``energy'' of the CWT of the time series $f$ at this particular scale. It allows for the identification of the most representative scales of a time series, i.e. the scales that contribute most to its total energy. Moreover, given a function $f\in L^2\left( \mathbb{R}\right)$ it can be decomposed into a sum of CWTs computed at base 2 power scales
\begin{equation}
\label{eq:fjk}
f=\sum _{j,k\in \mathbb{Z}}Wf\left( 2^k j,2^k\right)\psi _{2^k j,2^k},
\end{equation}
under some conditions on the wavelet (see \cite{mal98}). So, with respect to the representation of the scalogram and taking into account (\ref{eq:fjk}), it is convenient to use a binary logarithmic re-scalation in the abscissa axis, making base 2 power scales equidistant (see \cite{Bol17}).

\subsection{The scale index}

In this section, we define the \textit{scale index}, a
wavelet tool for measuring non-periodicity through the scalogram, which was previously
introduced in \cite{Ben10}. However, first we need some previous results.

The next theorem ensures that if a function $f$ has details at every scale (i.e. the scalogram of $f$ does not vanish at any scale), then it is non-periodic. For a detailed proof see \cite{Ben10}.

\begin{theorem}
\label{teo:1}
Let $\psi $ be a compactly supported wavelet. If $f:\mathbb{R}\rightarrow \mathbb{C}$ is a $T$-periodic function in $L^2\left( \left[ 0,T\right] \right) $, then $Wf\left( u,2T\right) =0$ for all $u\in \mathbb{R}$.
\end{theorem}

Note that if $f$ is a $T$-periodic function in $L^2\left(
\left[ 0,T\right] \right) $, and $\psi $ is a compactly supported wavelet, then $Wf\left(
u,s\right) $ is well-defined for $u\in \mathbb{R}$ and $s>0$, although $f$ is not in
$L^2\left( \mathbb{R}\right) $.

In Theorem \ref{teo:1}, $f$ is defined over the entire $\mathbb{R}$, but, in practice, time series have a finite time domain and produce ``border effects'' when the wavelet exceeds the domain limits. So, given a $T$-periodic time series with finite time domain, these border effects cause the scalogram not to vanish at scale $2T$. Therefore, we are going to define a special type of scalogram for precisely avoiding the undesired border effects.

\begin{definition}
\label{def:innersc}
Given a time series $f$ defined over a finite time interval $I=\left[ a,b\right]$, the \textit{inner scalogram} of a $f$ at a scale $s$ is given by
\begin{equation*}
\mathcal{S}^{\textrm{inner}}\left( s\right) :=\| Wf\left( s,u\right) \| _{J(s)}
=\left( \int _{c(s)}^{d(s)} | Wf\left( s,u\right) |^2 \textrm{d}u\right) ^{\frac{1}{2}},
\end{equation*}
where $J(s)=\left[ c(s),d(s) \right]\subseteq I $ is the maximal subinterval in $I$ for which the support of $\psi _{u,s}$ is included in $I$ for all $u\in J(s)$. Obviously, the length of $I$ must be big enough for $J(s)$ not to be empty or too small, i.e.  $b-a\gg sl$, where $l$ is the length of the support of $\psi $.

Since the length of $J(s)$ depends on the scale $s$, the values of the inner scalogram at
different scales cannot be compared. To avoid this problem, we can \textit{normalize} the
inner scalogram:
\begin{equation*}
\overline{\mathcal{S}}^{\textrm{inner}}\left( s\right) =
\frac{\mathcal{S}^{\textrm{inner}} \left( s\right) }{(d(s)-c(s))^{\frac{1}{2}}}.
\end{equation*}
\end{definition}

From Theorem \ref{teo:1} and Definition \ref{def:innersc} we obtain the following corollary, that can be also found in \cite{Ben10}.

\begin{corollary}
\label{result}
Let $\psi $ be a compactly supported wavelet. If $f:I=\left[ a,b\right] \rightarrow \mathbb{C}$ is a $T$-periodic function in $L^2\left( \left[ a,a+T\right] \right) $, then the (normalized) inner scalogram of $f$ at scale $2T$ is zero.
\end{corollary}

\begin{remark}
The concept of inner scalogram can be extended to wavelets that do not have compact support like Morlet wavelets, taking a compact ``effective support'' outside of which the magnitude of the wavelet can be considered negligible. But in this case, some theoretical results like Corollary \ref{result} may not hold exactly. According to \cite{tor98}, the effective support is defined by the $e$-folding time for the autocorrelation of wavelet power spectrum at each scale $s$ that, for Morlet wavelets, is equal to $\sqrt{2}s$.
\end{remark}

These results constitute a valuable tool for detecting periodic and non-periodic signals as a signal with details at every scale must be non-periodic. Moreover, since the scalogram of a $T$-periodic signal vanishes at all $2kT$ scales (for all $k\in \mathbb{N}$), it is sufficient to analyze only scales greater than a fundamental scale $s_0$. Thus, a signal which has details at arbitrarily large scales will be non-periodic.


On the other hand, we can ask if a function $f$ with $Wf\left( u,2T\right) =0$ for all
$u\in \mathbb{R}$ is really a $T$-periodic function. In general, it is not proved and it
remains as an open question. Nevertheless, this result can be proved in the particular case of Haar wavelets assuming that $f$ is bounded:

\begin{proposition}
\label{prop2.1}
Let $\psi $ be the Haar wavelet and let $f:\mathbb{R}\rightarrow \mathbb{R}$ be a continuous and bounded function. Given $T>0$, if $Wf\left( u,2T\right) =0$ for all $u\in
\mathbb{R}$, then $f$ is $T$-periodic.
\end{proposition}

\begin{proof}
Without loss of generality we may assume that $T=1$. Let us define
\[
g(u):=\int_{u}^{u+1} f(t)\,\textrm{d}t
\]
for all $u\in \mathbb{R}$. Then, using the Haar wavelet
\[
\psi (t) =\left\{
\begin{array}{ll}
1 & 0\leq t <1/2 \\
-1 & 1/2\leq t <1 \\
0 & \text{otherwise},
\end{array}
\right.
\]
we have
\[
Wf\left( u,2\right) =\int_{u}^{u+1} f(t)\,\textrm{d}t - \int_{u+1}^{u+2}
f(t)\,\textrm{d}t=0 \Longrightarrow g(u)=g(u+1)
\]
for all $u\in \mathbb{R}$ and hence $g$ is $1$-periodic.

We are going to prove that $g$ is in fact a constant function by reductio ad absurdum.
Supposing $g$ is not constant, we can find $u_1,u_2\in \left[ 0,1\right[ $ with
$c_1:=g\left( u_1\right) \neq c_2:=g\left( u_2\right) $. If $M$ is an upper bound for
$\left| f\right| $, then we take $N\in \mathbb{N}$ such that $N\left| c_2-c_1\right| >2M$.
Thus,
\begin{equation}
\label{eq:nc1}
\int _{u_1}^{u_1+N} f(t)\,\textrm{d}t=\sum _{j=1}^N \int _{u_1+j-1}^{u_1+j}
f(t)\,\textrm{d}t=\sum _{j=1}^N g\left( u_1\right) =Nc_1,
\end{equation}
since $g$ is $1$-periodic. Analogously,
\begin{equation}
\label{eq:nc2}
\int _{u_2}^{u_2+N} f(t)\,\textrm{d}t=Nc_2.
\end{equation}
So, from \eqref{eq:nc1} and \eqref{eq:nc2}, we have
\begin{equation}
\label{eq:gt2m}
\left| \int _{u_2}^{u_2+N} f(t)\,\textrm{d}t-\int _{u_1}^{u_1+N} f(t)\,\textrm{d}t\right|
=N\left| c_2-c_1\right| >2M.
\end{equation}
On the other hand, supposing $u_1<u_2$, we have
\begin{eqnarray}
\nonumber
\left| \int _{u_2}^{u_2+N} f(t) \, \textrm{d}t-\int _{u_1}^{u_1+N} f(t) \,
\textrm{d}t\right| = \left| \int _{u_2}^{u_1+N} f(t) \, \textrm{d}t+\int _{u_1+N}^{u_2+N}
f(t) \, \textrm{d}t -\int _{u_1}^{u_2} f(t)\,\textrm{d}t \right. \\
\label{eq:lt2m}
-\left.  \int _{u_2}^{u_1+N} f(t) \, \textrm{d}t\right| \leq \int _{u_1+N}^{u_2+N} \left|
f(t)\right| \, \textrm{d}t+\int _{u_1}^{u_2} \left| f(t)\right| \, \textrm{d}t \leq
2(u_2-u_1)M<2M.
\end{eqnarray}
Since \eqref{eq:gt2m} contradicts \eqref{eq:lt2m}, $g$ must be a constant function and then,
\[
0=g'(u)=\dfrac{\textrm{d}}{\textrm{d}u} \int_{u}^{u+1} f(t)\,\textrm{d}t =f(u+1)-f(u)
\]
for all $u\in\mathbb{R}$, concluding that $f$ is a $1$-periodic function.
\end{proof}

Taking into account these results, we are going to define the \textit{scale index}, that was previously introduced in \cite{Ben10}. However, we present a slightly different definition.

\begin{definition}
\label{def:si}
The \textit{scale index} of a time series $f$ in the scale interval $\left[ s_0,s_1\right] $ is given by the quotient
\begin{equation}
\label{eq:si}
i_{\textrm{scale}}:=\frac{\mathcal{S}(s_{\textrm{min}})}{\mathcal{S}(s_{\textrm{max}})},
\end{equation}
where $s_{\textrm{max}}\in \left[ s_0,s_1\right] $ is the smallest scale such that $\mathcal{S}(s)\leq \mathcal{S}(s_{\textrm{max}})$ for all $s\in \left[ s_0,s_1\right] $, and $s_{\textrm{min}}\in \left[ s_{\textrm{max}},2s_1\right] $ is the smallest scale such that $\mathcal{S}(s_{\textrm{min}})\leq \mathcal{S}(s)$ for all $s\in \left[ s _{\textrm{max}}, 2s_1 \right] $. Note that for compactly supported signals, the use of the normalized inner scalogram is recommended in order to fulfill Corollary \ref{result} and to avoid border effects.
\end{definition}

From its definition, $i_{\textrm{scale}}\in \left[ 0,1\right] $ and it can be interpreted as a measure of the degree of non-periodicity of a signal in the scale interval $\left[ s_0,s_1\right] $: the scale index will be numerically close to $0$ for periodic and quasi-periodic signals, and close to $1$ for highly non-periodic chaotic signals.

The difference between Definition \ref{def:si} and the definition given in \cite{Ben10} is that in Definition \ref{def:si} the scale $s_{\textrm{min}}$ is sought in $\left[ s_{\textrm{max}},2s_1\right] $ instead of $\left[ s_{\textrm{max}},s_1\right] $. This change is motivated by Theorem \ref{teo:1} and Corollary \ref{result}, because given a $T$-periodic signal (being $T$ the least period) its most representative scale is usually near $T$, and the scalogram vanishes at scale $2T$. So, we need to compute the scalogram up to the scale $2s_1$ in order to detect periodicities in $\left[ s_0,s_1\right] $. As a consequence of this fact, if there are scales in $\left[ s_1,2s_1\right] $ at which the scalogram takes a greater value than at $s_{\textrm{max}}$, then $s_1$ should be increased if we want to take into account these important scales in our analysis. Or, from another point of view, taking the definition of the scale index given in \cite{Ben10}, if the scale $s_{\textrm{max}}$ is in $\left[ s_1/2,s_1\right] $, then we should take a bigger $s_1$, at least up to $2s_{\textrm{max}}$, but only if we have no objection to include larger scales in our study of non-periodicity.

For example, Figure \ref{fig:senoscale} plots the scalogram and the scale index of a given signal, taking $s_0=4$ and varying $s_1$ from $4$ to $64$. It can be seen that $i_{\textrm{scale}}=1$ for $s_1\in \left[ 4,5.5\right] \cup \left[ 13.3,17\right] $ approximately, but there are scales in $\left[ s_1,2s_1\right] $ at which the scalogram takes a greater value than at $s_{\textrm{max}}$. So, as it is indicated above, if we want to make a wider scale study of non-periodicity, then $s_1$ should be increased in these cases. Nevertheless, if we are interested in measuring non-periodicity only in a determined scale interval $\left[ s_0,s_1\right] $ with $s_1\in \left[ 4,5.5\right] \cup \left[ 13.3,17\right] $, then we can conclude that the signal is highly non-periodic in $\left[ s_0,s_1\right] $, and the reason in this case is that there are significant scales outside the range of study.

\begin{figure}[tb]
\begin{center}
  \includegraphics[width=1\textwidth]{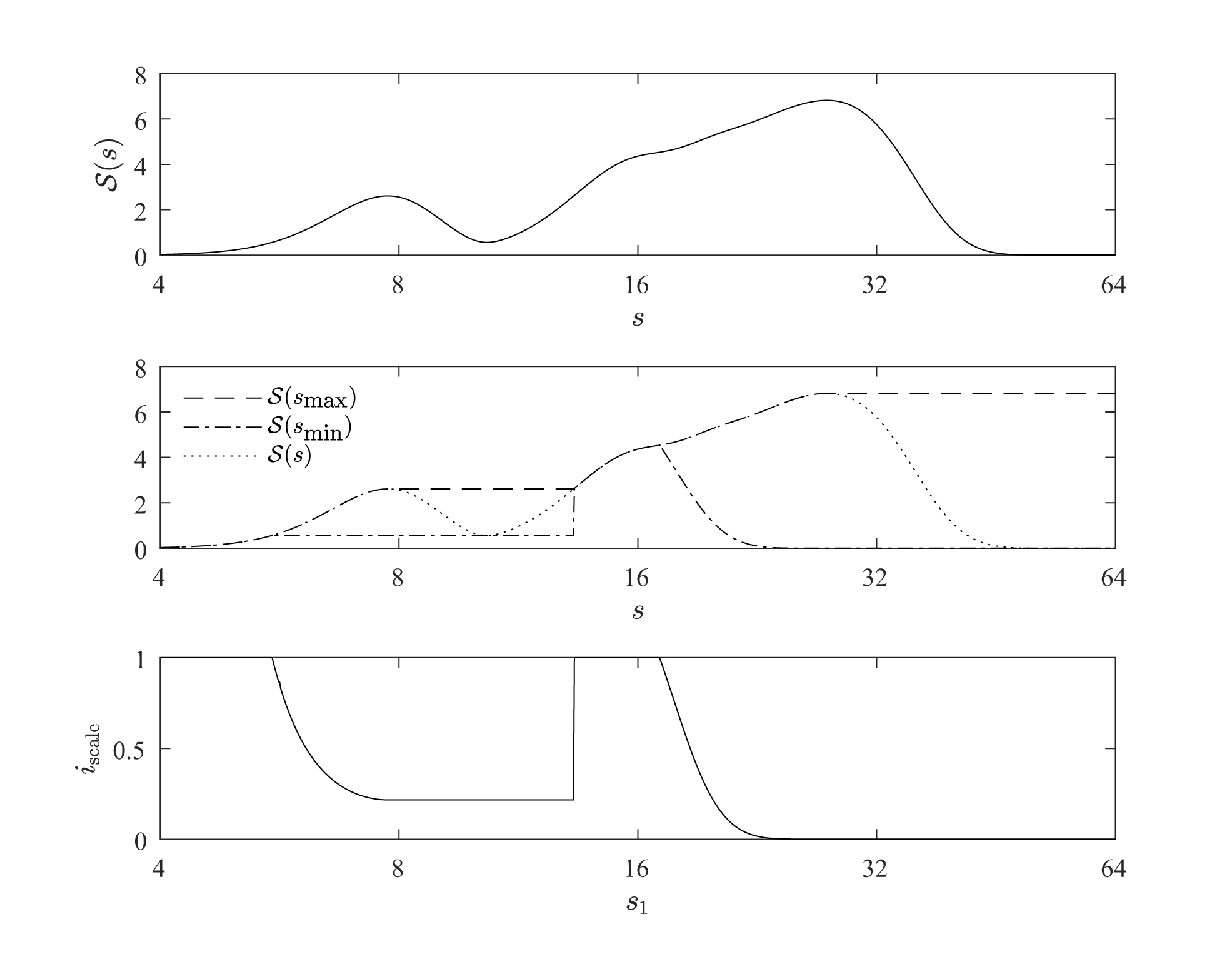}
\end{center}
\caption{Top: scalogram (using Morlet wavelet) of the signal $\sum _{i=1}^6 \sin\left( \frac{2\pi}{T_i}t\right) $ with $t\in \left[ 0,1000\right] $ and $T_i=8,16,20,24,28,32$ for $i=1,\ldots,6$ respectively, i.e. a combination of several signals of periods $T_i$. Center: the same scalogram as above (dotted), jointly with $\mathcal{S}\left( s_{\textrm{max}}\right) $ (dashed) and $\mathcal{S}\left( s_{\textrm{min}}\right) $ (dash-dotted) taking $s_0=4$ and $s_1=s$. Bottom: the corresponding scale index (i.e. $\mathcal{S}\left( s_{\textrm{min}}\right) /\mathcal{S}\left( s_{\textrm{max}}\right) $) taking $s_0=4$ and $s_1$ varying from $4$ to $64$.  }
\label{fig:senoscale}
\end{figure}

At this point, it is clear that the choice of the scale interval $\left[ s_0,s_1\right] $ is an important issue in the non-periodicity analysis. Since the non-periodic character of a signal is given by its behaviour at large scales, there is no need for $s_0$ to be very small. In general, we can choose $s_0$ such that $s_{\textrm{max}}=s_0+\epsilon$ where $\epsilon$ is positive and small. On the other hand, $s_1$ must be greater than all the relevant scales that we want to study. Moreover, if $s_{\textrm{min}}\simeq 2s_1$, then the scalogram decreases at large scales and it is recommended to increase $s_1$ in order to distinguish between a non-periodic signal and a periodic signal with a very large period. In general, it is recommended to make an study varying $s_1$ as in Figure \ref{fig:senoscale} (bottom) in order to visualize the evolution of the scale index.

\begin{remark}
It is known that the wavelet power spectra and, consequently, the scalogram of a signal are biased in favor of large scales (see \cite{Liu07}). This feature has to be taken into account by some wavelet tools that quantify the ``energy density'' based on the scalogram, such as the \textit{wavelet squared coherency} (see \cite{tor98,tor99}) or the \textit{windowed scalogram difference} (see \cite{Bol17}). In these tools, the scalogram is multiplied by the factor $\frac{1}{\sqrt{s}}$ for normalizing the weight of each scale (see Figure \ref{fig:sqrts}), or, from another point of view, the mother and daughter wavelets are normalized using the $L^1$-norm instead of the $L^2$-norm. However, in the case of the scale index, this correction would not be necessary because the scalogram of a white noise signal is more or less constant at all scales giving a scale index close to $1$ for any $s_1$, and this is the property that we want to preserve. If, on the other hand, we apply the correction for converting the scalogram into an ``energy density'' measure, the scale index of a white noise signal would tend to zero as we increase $s_1$ (see Figure \ref{fig:sqrts2}) and this is not desirable.
\end{remark}

\begin{figure}[tb]
\begin{center}
  \includegraphics[width=1\textwidth]{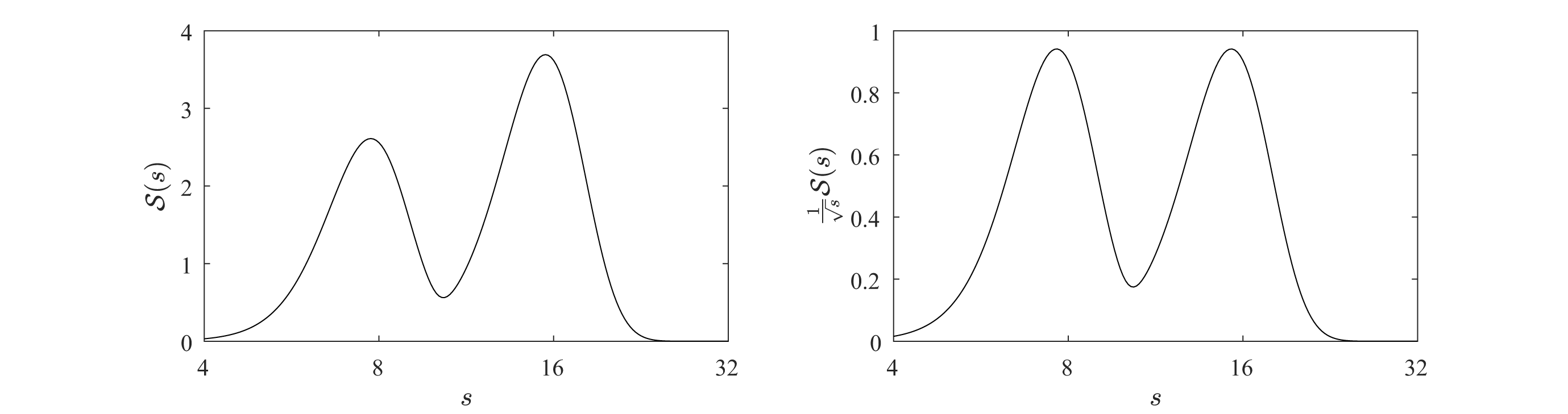}
\end{center}
\caption{Left: scalogram (using Morlet wavelet) of the signal $\sin\left( \frac{2\pi}{8}t\right) +\sin\left( \frac{2\pi}{16}t\right) $ with $t\in \left[ 0,1000\right] $, i.e. a combination of two signals of period $8$ and $16$ respectively with the same amplitude. There are two local maxima at the scales corresponding with the periods $8$ and $16$ respectively, but the scalogram takes different maximum values at these scales. Right: the same scalogram multiplied by the factor $\frac{1}{\sqrt{s}}$. Now, the scalogram takes the same maximum values, showing that both scales contribute with the same ``energy''.}
\label{fig:sqrts}
\end{figure}

\begin{figure}[tb]
\begin{center}
  \includegraphics[width=1\textwidth]{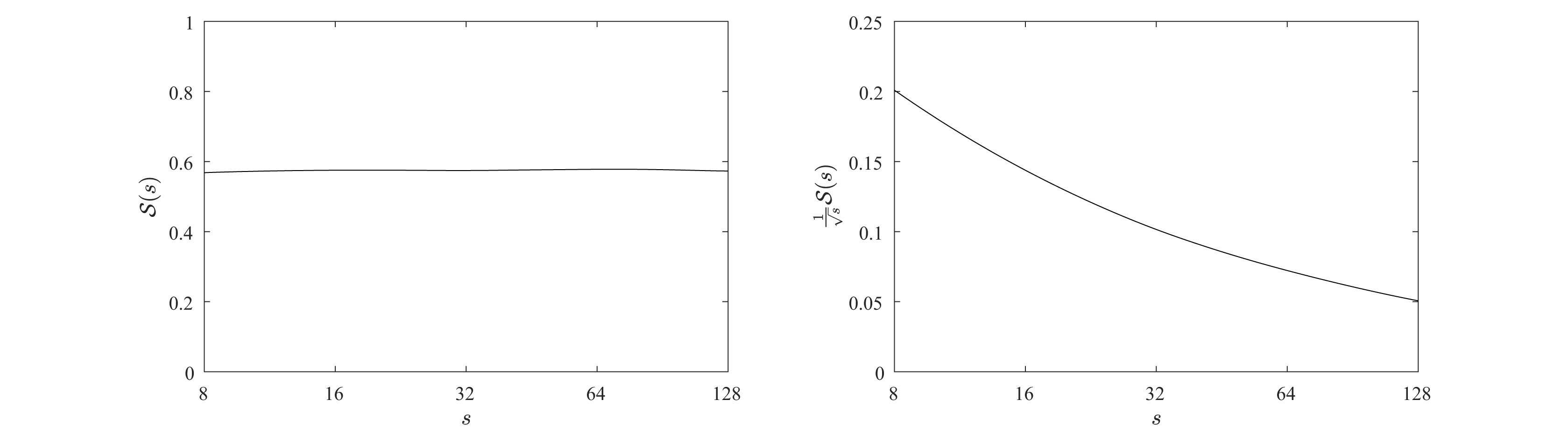}
\end{center}
\caption{Left: mean of the scalograms of $100$ random uniform series in $\left[ -1,1\right] $ of $10000$ data points, using Morlet wavelet. It can be seen that it is more or less constant and so, the corresponding scale index would be close to $1$ for any $s_1$. Right: mean of the $100$ scalograms of the same random uniform series but multiplied by the factor $\frac{1}{\sqrt{s}}$. It tends to zero for large scales and so, the scale index computed from this modified scalogram would also tend to zero as we increase $s_1$.}
\label{fig:sqrts2}
\end{figure}

\section{The windowed scale index}
\label{sec:wsi}

The scale index given by \eqref{eq:si} provides an estimate of the degree of non-periodicity of an entire time series. However, non-stationary signals may pass through different stages in which their behaviour varies considerably and therefore their non-periodicity may not be constant over time. For these cases, we introduce the next definitions.

\begin{definition}
The \textit{windowed scalogram} of a time series $f$ centered at time $t$ with time radius $\tau $ is given by
\begin{equation}
\label{wsc}
\mathcal{WS}_{\tau }(t,s) := \left( \int _{t-\tau } ^{t+\tau } | Wf\left( u,s\right) | ^2 \, \textrm{d}u \right)^{1/2}.
\end{equation}
The windowed scalogram was previously introduced in \cite{Bol17} and it is just the scalogram given by (\ref{scs}) restricted to a finite time interval $\left[ t-\tau ,t+\tau \right] $. 
\end{definition}

\begin{definition}
The \textit{windowed scale index} of a time series $f$ in the scale interval $[s_0,s_1]$ centered at time $t$ with time radius $\tau $ is defined as
\begin{equation}
wi_{\textrm{scale},\tau}(t) := \frac{\mathcal{WS}_{\tau} (t,s_{\textrm{min}})}{\mathcal{WS}_{\tau}(t,s_{\textrm{max}})},
\end{equation}
where, analogously to Definition \ref{def:si}, $s_{\textrm{max}}$ is the smallest scale such that $\mathcal{WS}_{\tau}(t,s)\leq \mathcal{WS}_{\tau}(t,s_{\textrm{max}})$ for all $s\in [s_0,s_1]$, and $s_{\textrm{min}}$ is the smallest scale such that $\mathcal{WS} _{\tau} (t,s _{\textrm{min}}) \leq \mathcal{WS} _{\tau} (t,s)$ for all $s \in \left[ s _{\textrm{max}}, 2s_1 \right] $.
\end{definition}

Although in the computation of the scale index it is recommended the use of normalized inner scalograms in order to avoid border effects, this recommendation is less important in the case of the windowed scale index. This is because for long series and relatively small time radii there would be no border effects in most of the windowed scalograms.

As the original scale index, the windowed scale index is also in the interval $[0,1]$ and it inherits the properties of the scale index that make it a tool for measuring the degree of non-periodicity of a signal around a given time. Precisely, this time dependence makes it a suitable tool for the study of non-stationary signals.

With respect to the choice of the time radius $\tau $, it depends on the nature of the time series $f$ and the objective of the study. In general, it must be big enough so that the windowed scalograms conserve desired (not noisy) information about $f$, but a too large $\tau $ produces inaccuracy in the location of events for non-stationary time series. Usually, if $\left[ a,b\right] $ is the support of $f$, then $\tau =\left( b-a\right) /20$ is a good choice.

\section{Examples and applications}
\label{sec:ex}

In this section we will illustrate the application of the windowed scale index to different scenarios of non-stationary time series such as chaotic time series, signals affected by non-stationary random noises and real time series coming from finance and economy which present both chaos and noise.

Computations were carried out using a self-developed R package \cite{Rsoft}, \texttt{wavScalogram}\cite{wavscalogram}, which can be freely downloaded from CRAN.

\subsection{The Bonhoeffer-van der Pol oscillator}

As it was done in \cite{Ben10}, the applicability of the windowed scale index is illustrated by studying some time series generated by the Bonhoeffer-van der Pol oscillator (BvP), which is given by the following non-autonomous planar system
\begin{equation}
\label{eq:bvdp}
\left. \begin{array}{rcl}
x'&=&x-\displaystyle{\frac{x^3}{3}}-y+I(t) \\
y'&=&c(x+a-by)\\
\end{array} \right\} ,
\end{equation}
being $a=0.7$, $b=0.8$, $c=0.1$, and $I(t)=A\cos \left( 2\pi t\right) $ an external periodic force. This oscillator is employed for modelling cardiac pulse (see \cite{Scot77}), and hence it is very useful to examine the non-periodicity of a signal without loosing the time location.

Taking the amplitude of the external force $A$ as a parameter, we can construct time series by means of the $x$ coordinate of the solution of \eqref{eq:bvdp} with initial conditions $(0,0)$. Specifically, we have created two time series for $A=0.74$ and $A=0.76$, with $t$ from $20$ to $400$ using a time step of $\Delta t=0.05$, which are called ``series 1'' and ``series 2'', respectively. Moreover, we have constructed another time series, designated by ``series 3'', using $A=0.74$ for $t\in \left[ 20,400\right] $ and $A=0.76$ for $t\in \left] 400,780\right] $. The normalized inner scalograms of these series are displayed in Figure \ref{figescint074076all}, using Daubechies eight--wavelet function and scales running from $0.05$ to $12.8$.

\begin{figure}[tbp]
\centering
\includegraphics[width=1\textwidth]{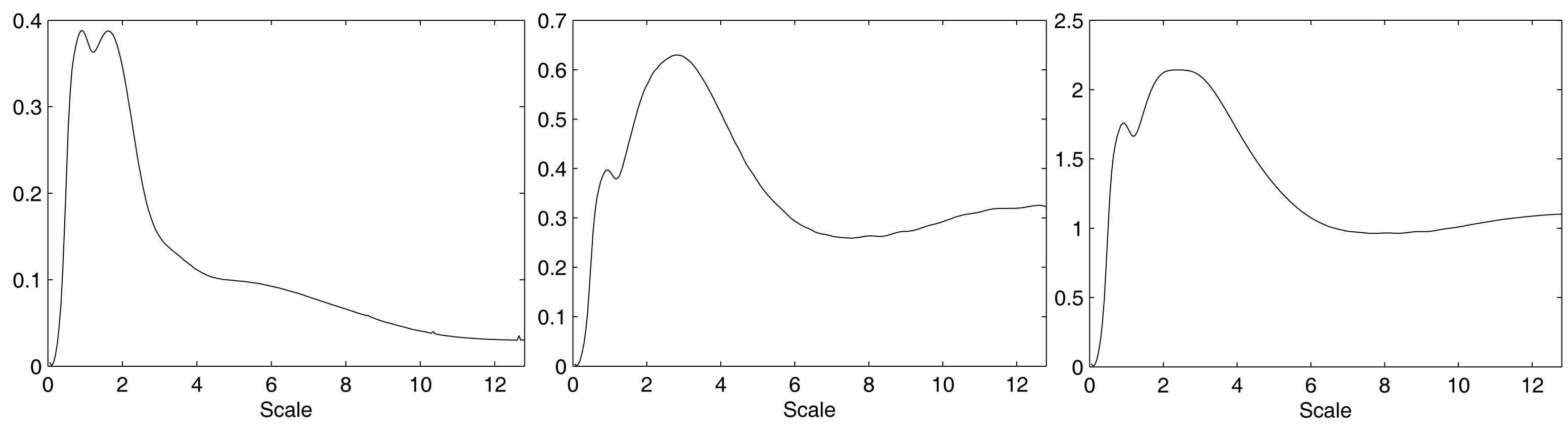}
\caption{Normalized inner scalograms of series 1 (left), 2 (center) and 3 (right), using Daubechies eight--wavelet. The scale runs from $0.05$ to $12.8$.}
\label{figescint074076all}
\end{figure}

For the computation of the windowed scale indices of these series (see Figure \ref{figsi074076all}), we have taken a time radius $\tau =50$ (i.e. $1000$ data points), $s_0 = 0.05$, $s_1 = 6.4$ and windowed normalized inner scalograms are utilized. This means that we have not taken into account CWTs where the corresponding wavelet overlaps the extremes of the time domain ($\left[ 20,400\right] $ for series 1 and 2, and $\left] 20,780\right] $ for series 3), and the scalograms are normalized as it is shown in Definition \ref{def:innersc}. As a curiosity, the global scale index of series 3 is approximately $0.4494$, which is greater than the scale index of series 1 ($0.0768$ approx.) and series 2 ($0.4116$ approx.). It is clearly shown in Figure \ref{figsi074076all} that the windowed scale index is more appropriate for studying non-stationary signals as it allows one to observe the evolution of the scale index over time.

\begin{figure}[tbp]
\centering
\includegraphics[width=1\textwidth]{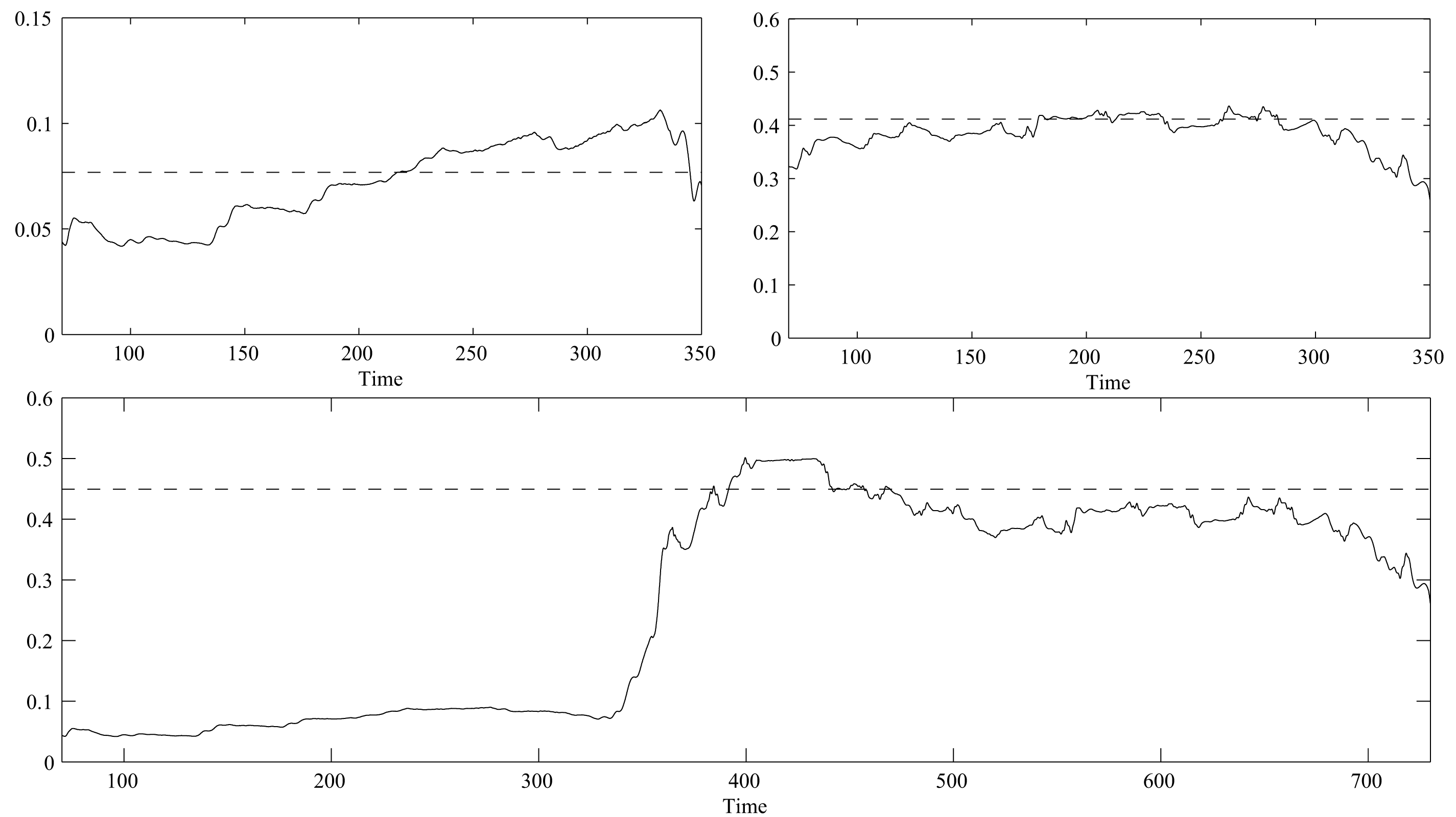}
\caption{Windowed scale indices of series 1 (up left), 2 (up right) and 3 (bottom) centered at different times with time radius $\tau =50$. The scale parameters run from $s_0=0.05$ to $s_1=6.4$ and we use Daubechies eight--wavelet function. The corresponding global scale index is represented with a dashed line.}
\label{figsi074076all}
\end{figure}

Finally, Figure \ref{figbvp074076_all_256} depicts the different stages for obtaining the windowed scale indices of series 3 for a wide set of values of $s_1$, ranging from $0.4$ to $6.4$. In this way, we can see the importance of the choice of $s_1$. Note that Figure \ref{figsi074076all} (bottom) only plots the windowed scale indices for $s_1 = 6.4$.

\begin{figure}[tbp]
\centering
\includegraphics[width=1\textwidth]{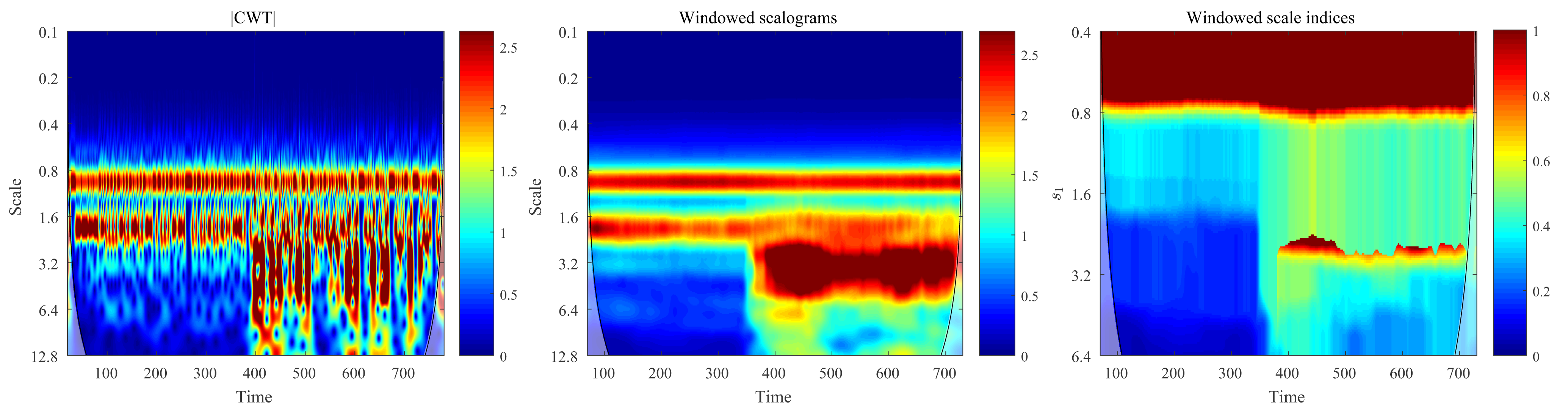}
\caption{Different stages for obtaining the windowed scale indices of series 3. Left: moduli of the CWT. Center: windowed scalograms centered at different times with time radius $\tau =50$. Right: windowed scale indices centered at different times with time radius $\tau =50$, taking $s_0=0.05$ and different values of $s_1$ from $0.4$ to $6.4$. The cone of influence outside which there are border effects is also represented with a black line in all plots.}
\label{figbvp074076_all_256}
\end{figure}

\subsection{Signals with increasing noise}
The windowed scale index is not only useful to study chaotic signals, as the one seen in the above example, but also it allows us to analyse other non-stationary signals such as those affected with non-stationary random noises. For example, we may consider signals of the form
\begin{equation}
\label{eq:elmar}
f(t)=\sin\left( 2\pi t\right) + \epsilon_t,\qquad \epsilon_t \sim \mathcal{N}\left( \mu =0, \sigma ^2=0.1t\right) ,
\end{equation}
which correspond to a $1$-periodic signal plus a Gaussian noise term $\epsilon_t$ with an increasing variance. In Figure \ref{figelmar} it is represented the average of the windowed scale indices of $1000$ series of the form \eqref{eq:elmar} taking $s_0=0.02$ and different $s_1$ values, using Morlet wavelets and time radius $\tau =5$. Since there is a $1$-periodic term, $s_1$ must be greater or equal than $1$ in order to detect this periodicity. Hence, for $s_1\geq 1$, it can be seen that the windowed scale indices increase in time along with the noise. But choosing a too large $s_1$ has drawbacks: the larger $s_1$ is, the less pronounced this increase will be.

\begin{figure}[tbp]
\centering
\includegraphics[width=1\textwidth]{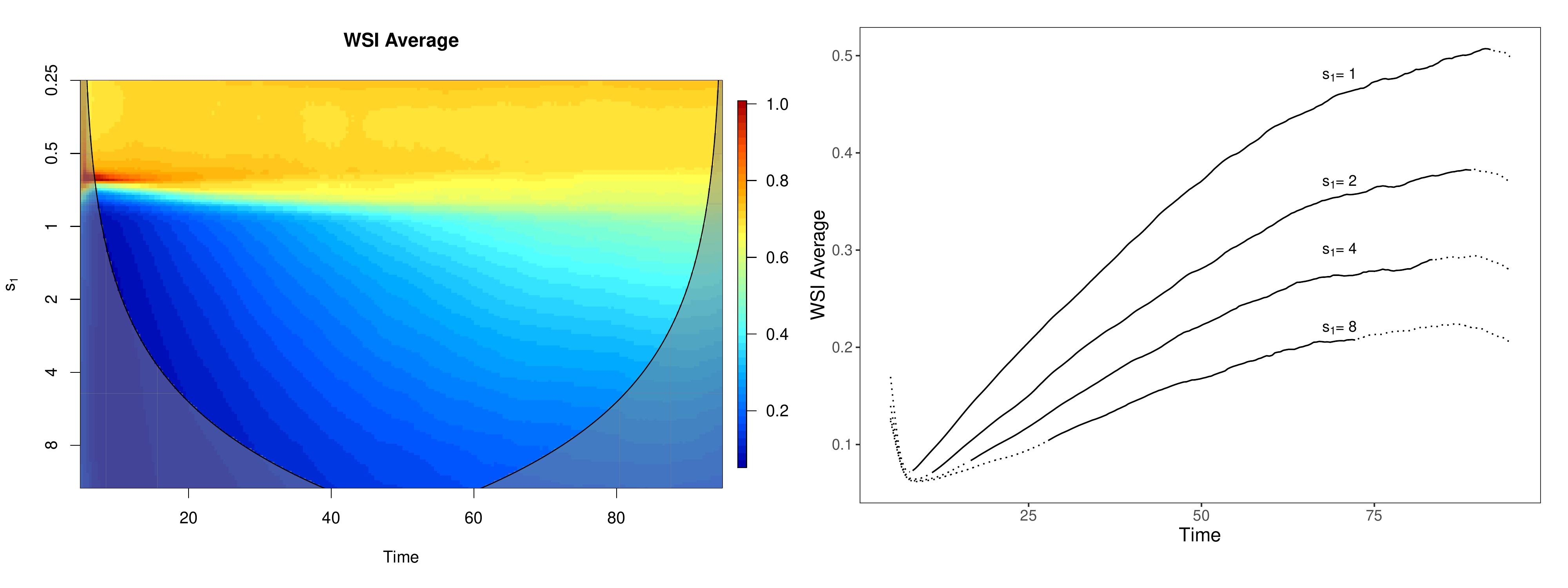}
\caption{Average of the windowed scale indices of $1000$ series of the form \eqref{eq:elmar} taking $s_0=0.02$ and different values for $s_1$, using Morlet wavelets and time radius $\tau =5$. The cone of influence outside which there are border effects is represented with a black line (left) and dots (right).}
\label{figelmar}
\end{figure}

\subsection{An economic application: crude oil and gold prices}

An interesting application of the windowed scale index in the field of economics consists of quantifying the degree of non-periodicity of the price of two major commodities such as crude oil and gold at different points in time and for certain ranges of scale. This analysis allows one to determine to what extent these time series have a periodical or repetitive behaviour.

Crude oil and gold are the most actively traded mineral commodities in the world and play a critical role in human civilization. As mentioned by \cite{Agu17}, oil is by far the most strategic commodity and the most sizable source of energy for the modern economy as it constitutes a vital input in the production process of many goods and services. The importance of crude oil is so great that changes in oil price affect strongly global economic growth, inflation and asset values. Meanwhile, gold is the leader in precious metal markets and is demanded for jewelry purposes and for many other industries such as dental, electronics and chemicals. Furthermore, gold has had a unique historic function as a means of exchange and a store of value for millennia and, more recently, as an alternative investment that hedges against inflation, a falling U.S. dollar or other forms of uncertainty. Due to its singular nature, gold is widely considered to be an effective safe haven asset against losses in financial markets, especially in times of financial and economic instability \cite{Bau10a,Bau10b,Cin13}.

The dataset used in this application consists of daily closing prices of futures contracts for gold and crude oil. Specifically, for oil the prices from nearby contracts of crude oil futures traded on the New York Mercantile Exchange (NYMEX) are utilized. For gold, the nearby contract prices of gold futures traded on the COMEX division of the NYMEX, which is the most liquid gold contract in the world, are employed. Changes in prices of oil and gold futures contracts are used in order to better capture the possible repetitive or periodical nature of these two series. The sample period spans from June 24, 1988 until June 20, 2017, with a total of 7564 daily observations. All data are collected from Thomson Reuters DataStream.

Figure \ref{oil_gold_wsc_si} displays the windowed scalograms and windowed scale indices of changes in the price of crude oil and gold futures contracts. The windowed scalograms use a time radius $\tau =300$ days and the windowed scale indices are calculated from these windowed scalograms taking $s_0=8$ days and a wide range of scales $s_1$, from $32$ to $512$ days. We have not used inner scalograms in this case. Therefore, the corresponding cones of influence outside which there are border effects have been plotted.

\begin{figure}[tbp]
\centering
\includegraphics[width=1\textwidth]{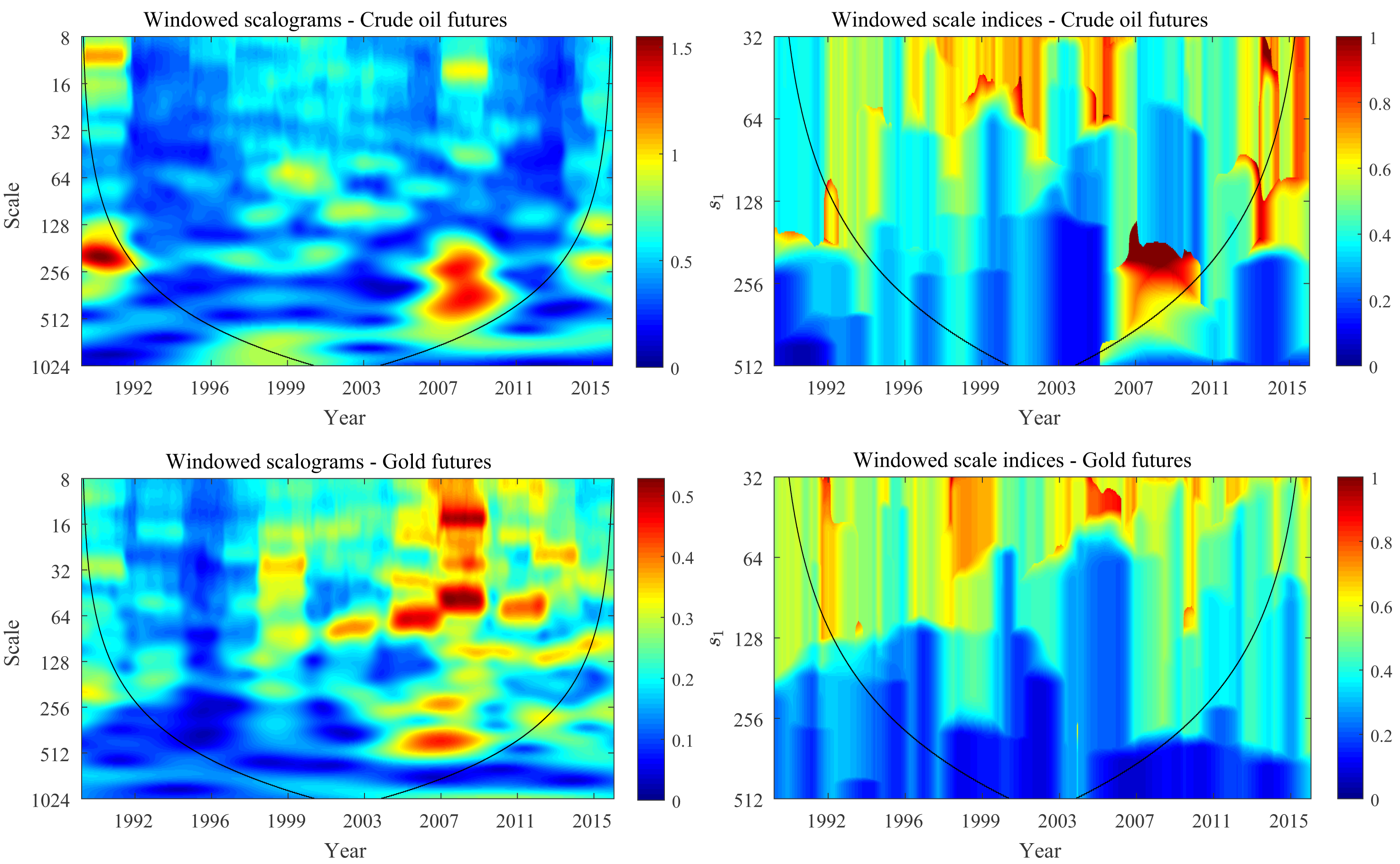}
\caption{Windowed scalograms and windowed scale indices of crude oil futures contracts (top) and gold futures contracts (bottom) price returns. Left: windowed scalograms centered at different times with time radius $\tau =300$ days. Right: windowed scale indices centered at different times with time radius $\tau =300$ days, taking $s_0=8$ days and different values of $s_1$ from $32$ to $512$ days. The cone of influence outside which there are border effects is also represented with a black line in all plots.}
\label{oil_gold_wsc_si}
\end{figure}

It is shown that the windowed scale indices for oil futures contracts take very high values in the time period between 2007 and 2010 for scale intervals with $s_1$ between $150$ and $300$ days (medium term). This finding may be closely related to an important shift in the behavior of oil price from 2007. In particular, as can be seen in the windowed scalograms for crude oil futures contracts in Figure \ref{oil_gold_wsc_si}, the scales between $256$ days (approximately one year) and $512$ days (two years) become the scales that contribute most to the energy of the oil price series during the period 2007-2010, and by taking $s_1$ less than $300$, these important scales are not taken into account. Hence, these high values of the windowed scale indices for oil futures can be interpreted as a clear symptom of unpredictability (see Section \ref{sec:unpred}) in oil price at the medium term during the period around the financial crisis that started in mid-2007 in the U.S. subprime mortgage market and spread throughout the world from September 2008 following the collapse of the U.S. investment bank Lehman Brothers. During this period, crude oil price experienced drastic rises and falls mostly driven by positive and negative shifts in global aggregate demand. For example, as noted by \cite{Rat13}, the rapid increase in the spot price per barrel of West Texas Intermediate (WTI) crude oil from \$58 in January 2007 to \$140 in June 2008 may be associated with a strong global economic activity. In turn, the fall in oil price to \$41.68 in January 2009 may be attributed to the global financial crisis and the subsequent great recession. However, the WTI oil price rebounded to \$133.93 in April 2011, even though global economic activity remained moderate. Therefore, the abnormally volatile development of crude oil price over the last years may have largely contributed to the high degree of unpredictability of oil price at the medium term around the recent global financial crisis, which is captured by the windowed scale indices.

In contrast, the windowed scale indices for gold futures contracts do not take generally high values despite the fact that during the financial crisis period there were also significant changes in the energies contributed by most of the scales, as is shown in the windowed scalograms plotted in Figure \ref{oil_gold_wsc_si}. This implies a more predictable behaviour of the price of gold for all time scales over the last few decades in comparison with oil prices. This result may be due to the fact that gold price has followed a secular upward trend, which has been only interrupted sporadically and for very short periods of time. For instance, unlike conventional financial assets such as stocks or bonds, gold price markedly increased during the recent global financial crisis. This evidence confirms that gold acts as a safe haven or a hedge in times of financial turbulence as its inclusion in traditional portfolios reduces investors' losses in the face of extreme negative market shocks.

\subsection{Non-periodicity and unpredictability}
\label{sec:unpred}

So far, we have showed how the windowed scale index can be used to assess the degree of non-periodicity of a time series over time and thus it has been proven useful for determining shifts in the periodicity regime. Since periodicity and predictability represent closely related concepts, it seems reasonable to compare the windowed scale index with a standard measure of the level of unpredictability of a time series.

The \textit{sample entropy} (SampEn) \cite{Joshua2000} is a measure of the signal's complexity and ultimately of its unpredictability, in the sense that the higher the SampEn is, the more unpredictable the signal will be. In essence, the SampEn is a modification of the \textit{approximate entropy} (AppEn) \cite{Pincus1991}, which was designed to measure the unpredictability of time series, but it overcomes some of the limitations of the AppEn, such as the signal length dependence. Both measures are based on the Kolmogorov-Sinai statistic \cite{Sinai59} and have been mainly used in the analysis of biomedical data (e.g. EEG signals). Similarly, the scale index (global or windowed) may also be seen as a measure of  a time series unpredictability in so far as the degree of non-periodicity intuitively gives us an insight into how the future values of a signal may be predicted from the past ones.

In order to compare the SampEn with the windowed scale index, the sample entropies are computed on a rolling window of the same radius as the one employed for the windowed scale indices estimation. The computations are performed using the R statistical software and the SampEn functions provided by the package \texttt{pracma} \cite{pracma}. For the artificial signals generated by the Bonhoeffer-van der Pol chaotic dynamical system, Figure \ref{figsen074076all} depicts the sample entropies for a rolling window of radius $\tau = 50$ observations and the global SampEn calculated for the whole signal. Although the overall value of the SampEn is slightly higher for time series 2, which corresponds to a more chaotic signal than time series 1, the difference is so small that it can be considered negligible. By contrast, when comparing Figure \ref{figsen074076all} with Figure \ref{figsi074076all}, it can be seen that the windowed scale index is able to clearly distinguish between different degrees of chaos or non-periodicities. 

\begin{figure}[tbp]
\centering
\includegraphics[width=1\textwidth]{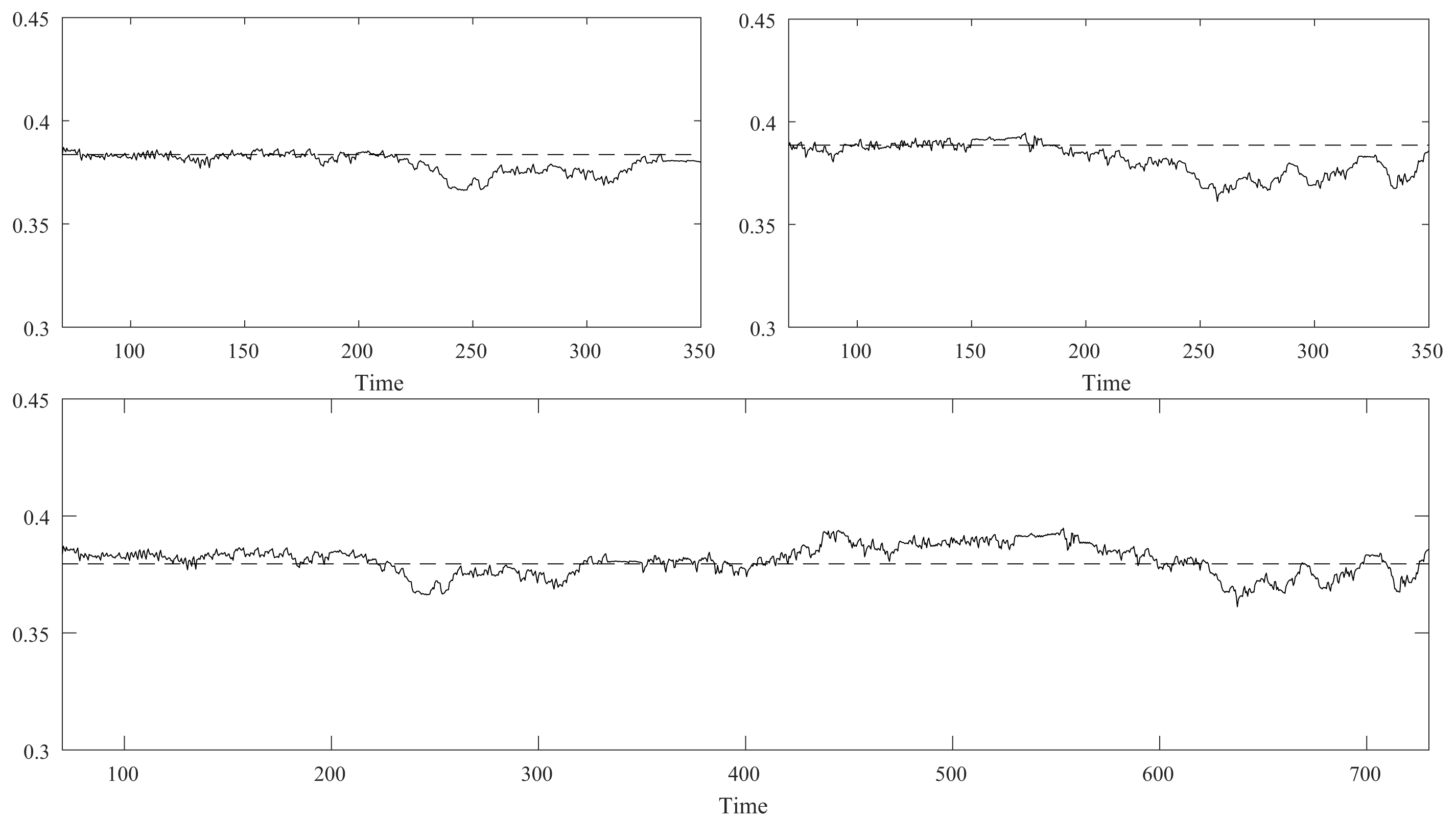}
\caption{Sample entropies of series 1 (up left), 2 (up right) and 3 (bottom) considering subseries centered at different times with time radius $\tau =50$. The corresponding global SampEn is represented with a dashed line.}
\label{figsen074076all}
\end{figure}

Likewise, the sample entropies are estimated for the crude oil and gold futures contracts in Figure \ref{figsenwsi_oilgold} (top), using a rolling window of radius $\tau =300$ days. The global SampEn for the oil futures contracts series yields overall higher values than for the gold futures contracts (dashed lines) but the differences are again very small. On the other hand, the sample entropies estimated on a rolling window (solid line) take, for certain time intervals, larger values for the crude oil futures than for the gold futures contracts, and vice versa. In order to compare with the windowed scale index, in Figure \ref{figsenwsi_oilgold} (bottom) we also show the sections at $s_1 =256$ days (a representative scale for medium term) of the windowed scale indices plots of Figure \ref{oil_gold_wsc_si}, in which a time radius of $\tau =300$ days was also employed. Moreover, the corresponding global scale index is represented by a dashed line. It can be seen that the windowed scale indices of crude oil futures contracts increase considerably during the global financial crisis, unlike what happens with the gold futures contracts. Besides, the global scale index of the gold futures contracts is slightly greater than the global scale index of the crude oil futures contracts, showing that in this case, the windowed scale index provides more reliable and detailed information than the global scale index.

\begin{figure}[tbp]
\centering
\includegraphics[width=1\textwidth]{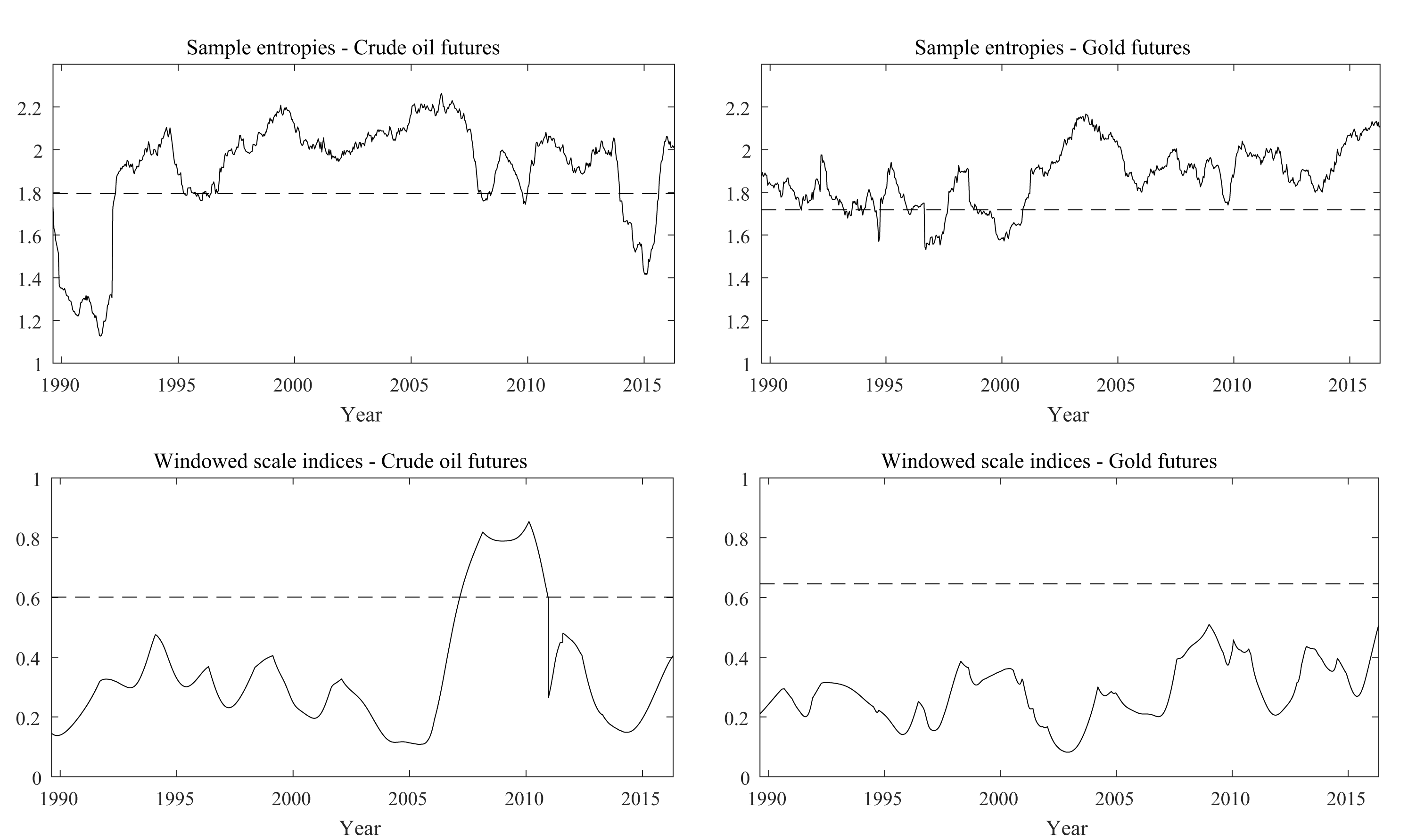}
\caption{Sample entropies and windowed scale indices of crude oil and gold futures contracts price returns. For the sample entropies (top), we have considered subseries centered at different times with time radius $\tau =300$ days. With respect to the windowed scale indices (bottom), they are centered at different times with time radius $\tau =300$ days, taking $s_0=8$ days and $s_1=256$ days. The corresponding global SampEn and scale index are represented with dashed lines.}
\label{figsenwsi_oilgold}
\end{figure}

\section{Conclusions}
\label{sec:conc}

The scale index proposed by \cite{Ben10} has proven to be a useful wavelet tool to assess the degree of non-periodicity of a time series. Such relevance is based on the result according to which the wavelet transform of a periodic time series vanishes when the scale is twice the period. From a theoretical point of view, it is very interesting to study whether the converse is true or not. In this work, it is shown that, at least for the Haar wavelet, the vanishing of the wavelet transform at a certain scale (for all times) guarantees the periodic character of the time series.

The windowed scale index introduced in this paper has its origin in the scale index and is particularly adequate to measure the degree of non-periodicity of non-stationary signals, whose non-periodicity may vary over time. In this regard, the windowed scale index can be used to detect changes in the predictability regime of time series, as it is evidenced in the economic application on crude oil and gold futures contracts.

In addition, our results suggest that the use of the windowed scale index as a measure of the unpredictability of a time series has several advantages over some standard unpredictability measures such as the SampEn or the AppEn. The first one is about the interpretability. Specifically, the windowed scale index takes values between 0 and 1 and this boundedness allows us to clearly state how non-periodic a time series is. Meanwhile, the SampEn and AppEn are not bounded, so it is harder to say whether a given entropy value is high or low and, consequently, whether a given time series is highly unpredictable or not. A second advantage of the windowed scale index over the entropy parameters is its sensitivity. In particular, the windowed scale index presents a higher sensitivity to distinguish between different levels of non-periodicity. For instance, while the windowed scale index is able to clearly identify that time series 1 in the BvP oscillator example is much less non-periodic than time series 2, the SampEn values are almost identical. Similarly, in the economic application the windowed scale index shows how during the recent global financial crisis the oil futures contracts exhibit a higher degree of non-periodicity than gold futures contracts. This finding seems reasonable taking into account that gold has followed a secular upward trend as it is traditionally considered a safe haven in times of financial turmoil. In contrast, the SampEn parameters provide almost undistinguishable values and when computed in a rolling window framework, they are unable to identify with clarity shifts in the pattern of unpredictability caused by the international financial crisis of 2007-2008.

\bibliography{manuscript}

\begin{thebibliography}{10}

\bibitem{Ben10}
R.~Ben\'{\i}tez, V.~J. Bol\'os, and M.~E. Ram\'{\i}rez.
\newblock A wavelet-based tool for studying non-periodicity.
\newblock {\em Computers \& Mathematics with Applications}, 60(3):634--641,
  August 2010.

\bibitem{Hes13}
M.~Hesham.
\newblock Wavelet-scalogram based study of non-periodicity in speech signals as
  a complementary measure of chaotic content.
\newblock {\em International Journal of Speech Technology}, 16(3):353--361,
  September 2013.

\bibitem{Akh14}
A.~Akhshani, A.~Akhavan, A.~Mobaraki, S.~C. Lim, and Z.~Hassan.
\newblock Pseudo random number generator based on quantum chaotic map.
\newblock {\em Communications in Nonlinear Science and Numerical Simulation},
  19(1):101--111, January 2014.

\bibitem{Ava15}
Erdin\c{c} Avaro\u{g}lu, Taner Tuncer, A.~Bedri \"Ozer, Burhan Ergen, and
  Mustafa T\"urk.
\newblock A novel chaos-based post-processing for {TRNG}.
\newblock {\em Nonlinear Dynamics}, 81(1-2):189--199, July 2015.

\bibitem{Yan16a}
Yu-Guang Yang, Peng Xu, Rui Yang, Yi-Hua Zhou, and Wei-Min Shi.
\newblock Quantum {Hash} function and its application to privacy amplification
  in quantum key distribution, pseudo-random number generation and image
  encryption.
\newblock {\em Scientific Reports}, 6:srep19788, January 2016.

\bibitem{Yan16b}
Yu-Guang Yang and Qian-Qian Zhao.
\newblock Novel pseudo-random number generator based on quantum random walks.
\newblock {\em Scientific Reports}, 6:srep20362, February 2016.

\bibitem{Tun16}
Taner Tuncer.
\newblock The implementation of chaos-based {PUF} designs in field programmable
  gate array.
\newblock {\em Nonlinear Dynamics}, 86(2):975--986, October 2016.

\bibitem{Mur16}
M.~A. Murillo-Escobar, C.~Cruz-Hern\'andez, L.~Cardoza-Avenda\~no, and
  R.~M\'endez-Ram\'{\i}rez.
\newblock A novel pseudorandom number generator based on pseudorandomly
  enhanced logistic map.
\newblock {\em Nonlinear Dynamics}, 87(1):407--425, January 2017.

\bibitem{Yan15}
Yu-Guang Yang, Qing-Xiang Pan, Si-Jia Sun, and Peng Xu.
\newblock Novel {Image} {Encryption} based on {Quantum} {Walks}.
\newblock {\em Scientific Reports}, 5:srep07784, January 2015.

\bibitem{Qib13}
Qibin Fan, Yanxin Wang, and Li~Zhu.
\newblock Complexity analysis of spatial–temporal precipitation system by
  {PCA} and {SDLE}.
\newblock {\em Applied Mathematical Modelling}, 37(6):4059--4066, March 2013.

\bibitem{Beh13}
Sohrab Behnia, Javid Ziaei, Marjan Ghiassi, and Mohammad Yahyavi.
\newblock Comprehensive chaotic description of heartbeat dynamics using scale
  index and lyapunov exponent.
\newblock In {\em Proceedings, 6th Chaotic Modeling and Simulation
  International Conference}, volume 500, pages 1--5, 2013.

\bibitem{Fel15}
Jorge Luis~Palacios Felix, Jos\'e~Manoel Balthazar, \'Angelo~Marcelo Tusset,
  Vin\'{\i}cius Piccirillo, Atila~Madureira Bueno, and Reyolando Manoel Lopes
  Rebello da~Fonseca Brasil.
\newblock On {Optimal} {Control} of a {Nonlinear} {Robotic} {Mechanism} {Using}
  the {Saturation} {Phenomenon}.
\newblock In {\em Structural {Nonlinear} {Dynamics} and {Diagnosis}}, Springer
  {Proceedings} in {Physics}, pages 145--165. Springer, Cham, 2015.

\bibitem{Pic15}
Vinicius Piccirillo, Jos\'e~M Balthazar, Angelo~M Tusset, Davide Bernardini,
  and Giuseppe Rega.
\newblock Characterizing the nonlinear behavior of a pseudoelastic oscillator
  via the wavelet transform.
\newblock {\em Proceedings of the Institution of Mechanical Engineers, Part C:
  Journal of Mechanical Engineering Science}, 230(1):120--132, January 2016.

\bibitem{Seu15}
Diego~Seuret Jim\'enez, P.~Stahl, and O.~Terminel.
\newblock Mechanical fault identification using {Wavelet} {Transform} and
  {Labview}.
\newblock {\em Nova Scientia}, 7(14):162--177, May 2015.

\bibitem{Bol17}
V.~J. Bol\'os, R.~Ben\'{\i}tez, R.~Ferrer, and R.~Jammazi.
\newblock The windowed scalogram difference: {A} novel wavelet tool for
  comparing time series.
\newblock {\em Applied Mathematics and Computation}, 312:49--65, November 2017.

\bibitem{mal98}
Stephane Mallat.
\newblock {\em A {Wavelet} {Tour} of {Signal} {Processing}: {The} {Sparse}
  {Way}}.
\newblock Academic Press, December 2008.

\bibitem{tor98}
Christopher Torrence and Gilbert~P. Compo.
\newblock A {Practical} {Guide} to {Wavelet} {Analysis}.
\newblock {\em Bulletin of the American Meteorological Society}, 79(1):61--78,
  January 1998.

\bibitem{Liu07}
Yonggang Liu, X.~San~Liang, and Robert~H. Weisberg.
\newblock Rectification of the {Bias} in the {Wavelet} {Power} {Spectrum}.
\newblock {\em Journal of Atmospheric and Oceanic Technology},
  24(12):2093--2102, December 2007.

\bibitem{tor99}
Christopher Torrence and Peter~J. Webster.
\newblock Interdecadal {Changes} in the {ENSO}–{Monsoon} {System}.
\newblock {\em Journal of Climate}, 12(8):2679--2690, August 1999.

\bibitem{Rsoft}
{R Core Team}.
\newblock {\em R: A Language and Environment for Statistical Computing}.
\newblock R Foundation for Statistical Computing, Vienna, Austria, 2017.

\bibitem{wavscalogram}
V.~J. Bolós and R.~Benítez.
\newblock {\em wavScalogram: Wavelet Scalogram Tools for Time Series Analysis},
  2019.
\newblock R package version 1.0.0.

\bibitem{Scot77}
Alwyn Scott.
\newblock {\em Neurophysics}.
\newblock Wiley, 1977.

\bibitem{Agu17}
R.~F. Aguilera and M.~Radetzki.
\newblock The synchronized and exceptional price performance of oil and gold:
  Explanations and prospects.
\newblock {\em Resources Policy}, 54:81--87, 2017.

\bibitem{Bau10a}
D.~G. Baur and B.~M. Lucey.
\newblock Is gold a hedge or a safe haven? an analysis of stocks, bonds and
  gold.
\newblock {\em The Financial Review}, 45:217--229, 2010.

\bibitem{Bau10b}
D.~G. Baur and T.~K. McDermott.
\newblock Is gold a safe haven? international evidence.
\newblock {\em Journal of Banking \& Finance}, 24:1886--1898, 2010.

\bibitem{Cin13}
C.~Ciner, C.~Gurdgiev, and B.~M. Lucey.
\newblock Hedges and safe havens: An examination of stocks, bonds, gold, oil
  and exchange rates.
\newblock {\em International Review of Financial Analysis}, 29:202--211, 2013.

\bibitem{Rat13}
R.~A. Ratti and J.~L. Vespignani.
\newblock Why are crude oil prices high when global activity is weak?
\newblock {\em Economic Letters}, 21:133--136, 2013.

\bibitem{Joshua2000}
Joshua~S. Richman and J.~Randall Moorman.
\newblock Physiological time-series analysis using approximate entropy and
  sample entropy.
\newblock {\em American Journal of Physiology-Heart and Circulatory
  Physiology}, 278(6):H2039--H2049, 2000.

\bibitem{Pincus1991}
Steven~M. Pincus, Igor~M. Gladstone, and Richard~A. Ehrenkranz.
\newblock A regularity statistic for medical data analysis.
\newblock {\em Journal of Clinical Monitoring}, 7(4):335--345, Oct 1991.

\bibitem{Sinai59}
Ya~G. Sinai.
\newblock On the notion of entropy for a dynamic system.
\newblock {\em Doklady of Russian Academy of Sciences}, 124(9):768--771, 1959.

\bibitem{pracma}
Hans~Werner Borchers.
\newblock {\em pracma: Practical Numerical Math Functions}, 2017.
\newblock R package version 2.0.7.

\end{thebibliography}
\bibliographystyle{unsrt}

\end{document}